\pdfoutput=1

\documentclass[preprint,11pt,3p]{elsarticle}

\usepackage[normalem]{ulem}

\usepackage[numbers]{natbib}
\usepackage{amssymb}
\usepackage{cancel}
\usepackage{graphicx}
\usepackage{float}
\usepackage{pgf,tikz}

\usetikzlibrary{snakes}
\tikzstyle{printersafe}=[snake=snake,segment amplitude=0 pt]

\usetikzlibrary{arrows}
\usetikzlibrary{shapes}
\usepackage{amsmath}
\usepackage{caption}
\usepackage{todonotes}
\usepackage[linesnumbered,ruled]{algorithm2e}
\usepackage{pgf, tikz}
\usetikzlibrary{graphs, quotes, graphs.standard}
\usetikzlibrary{matrix}
\usepackage{tabularx}
\usepackage{arydshln,leftidx,mathtools}
\usepackage{setspace}

\usepackage{comment}
\usepackage{mathtools}

{\endarray\hspace*{-0.5\arraycolsep}\endBmatrix}
{\endarray\endbmatrix}

\usepackage{hyperref}
\hypersetup{
    colorlinks=true,
    linkcolor=blue,
    filecolor=magenta,      
    urlcolor=cyan,
    pdfpagemode=FullScreen,
    }

\usepackage{subfiles}
\usepackage{enumerate}
\usepackage{bm}
\usepackage[T1]{fontenc}
\usepackage[utf8]{inputenc}
\usepackage[english]{babel}
\usepackage{float}

\usepackage{accents}

\newtheorem{claim}{Claim}

\newtheorem{theorem}{Theorem}
\newtheorem{proposition}{Proposition}

\newtheorem{obs}{Observation}
\newtheorem{definition}{Definition}

\newtheorem{lemma}{Lemma}
\newtheorem{corollary}{Corollary}

\newenvironment{proof}{ {\bf Proof:}} %{$\Box$}

\newcommand{\revision}[1]{{#1}}                   
\newcommand{\STAB}{{\hbox{STAB}}}
\newcommand{\R}{{\mathbb{R}}}
\newcommand{\proj}{{\hbox{proj}}}

\begin{document}
\onehalfspacing

\begin{frontmatter}

\title{%\sout{Facets of the Total Matching Polytope for bipartite graphs}
{The Total Matching Polytope of Complete Bipartite Graphs}} 

\author[1]{Yuri Faenza}
\ead{yf2414@columbia.edu}

\author[2]{Luca Ferrarini}%\corref{cor1}\fnref{label3}}
\ead{l.ferrarini3@campus.unimib.it}

%\institute{università di Pavia, Dipartimento di Matematica ``F. Casorati'' \and Columbia University, New York, Ny}
\address[1]{IEOR Department, Columbia University}
\address[2]{Università di Pavia, Dipartimento di Matematica ``F. Casorati''}

%\address[label4]{Università di Pavia, Dipartimento di Matematica ``F. Casorati''}

\begin{abstract}
The {\it total matching polytope} generalizes the stable set polytope and the matching polytope.
In this paper, we first propose new facet-defining inequalities for the total matching polytope. We then give an exponential-sized, non-redundant description in the original space and a compact description in an extended space of the total matching polytope of complete bipartite graphs. %for complete bipartite graphs. ofgive exact formulations for the total matching polytope for trees and bip

%

%In this paper, we give the perfect formulation for Trees and we derive two new families of valid inequalities, the {\it balanced biclique inequalities} which are always facet-defining and the {\it non-balanced lifted biclique inequalities} obtained by a lifting procedure, which are facet-defining for bipartite graphs.
%
%Finally, we give a complete description for Complete Bipartite Graphs.
%

\end{abstract}

\begin{keyword}
Integer Programming \sep Combinatorial Optimization  \sep Total Matching \sep Polyhedral Combinatorics \sep Complete Bipartite Graphs \sep Extended Formulations
\end{keyword}

\end{frontmatter}

\section{Introduction}
Let $G=(V,E)$ be a simple, loopless and undirected graph, and let $D = V \cup E$ be the set of its \emph{elements}, i.e., vertices and edges.
Elements $d,d' \in D$ are said to be {\it adjacent} if $d$ and $d'$ are adjacent vertices, or incident edges, or $d$ (resp., $d'$) is an edge incident to a vertex $d'$ (resp, $d$).
If $d,d' \in D$ are not adjacent, they are {\it independent}.
A \textit{stable set} is a set of pairwise independent vertices, while a \textit{matching} is a set of pairwise independent edges.
A \textit{total matching} is a set of pairwise independent elements. Hence stable sets and matching are total matchings, but the converse may not be true. For instance, the complete graph with $3$ nodes has a total matching of size $2$, while all its matchings and stable sets have size at most $1$.
%
%A subset $C \subseteq V \cup E$ is a \textit{total cover} of $G$ if it covers all the elements of $G$.
%
The \emph{total matching problem} asks for a total matching of maximum size.
%
%This problem generalizes both the Stable Set Problem, where we look for a stable set of maximum size and the Matching Problem, where instead we look for a matching of maximum size.
%
%In particular, a matching is called {\it perfect} if it covers all the vertices, that is, has size $\frac{1}{2}|V|$. 
%

Define $\nu_T(G) := \max \{ |T| : T \mbox{ is a total matching}\}$, $\nu(G):= \max \{ |M| : M \mbox{ is a matching}\}$ and $\alpha(G):= \max \{|S| : S \mbox{ is a stable set}\}$.
%%%%%%%%%%%%%%%%%%%%%%%%%%%%%%%%%
%
$\nu(G)$ and $\alpha(G)$ have been extensively studied in the literature, both algorithmically and in relation with the corresponding polytopes, see, e.g.,~\cite{aprile2020extended,Aprile201775,Chvatal,Edmonds1965,quasi-line,faenza2014solving,faenza2021separation,Yuri,Ventura,Letchford,Oriolo,Padberg1973,Rossi2001}. Despite the fact that it generalizes those illustrious special cases, and its connection with Vizing's long-standing \emph{total coloring conjecture} (see, e.g.,~\cite{Polyhedra} for details), the total matching problem is less studied in the operations research literature. 
In particular, significant results have been obtained only for very structured graphs, such as cycles, paths, full binary trees, hypercubes, and complete graphs~\cite{Leidner2012}.  The first works on the total matching problem appeared in \cite{TotalMatching,NordHaus}. %, where the authors derive lower and upper bounds on the size of a maximum total matching.
%In particular, the authors show that:\begin{center}    $\nu_{T}(G) \geq \max\{\alpha(G),\nu(G)\}$\end{center}
%In \cite{TotalMatching}, the authors find a relation between $\nu_{T}(G)$ and $\tau(G):=\min \{|C| : C \mbox{ is a total cover}\}$, indeed
%they show that:\begin{center}    $\tau(G) \leq \nu_{T}(G)$\end{center}
%
In \cite{Manlove}, Manlove %provides a survey of the algorithmic complexities of the decision problems related to the previous parameters.
%The author 
 reports that $\nu_{T}(G)$ can be computed in polynomial time for trees and it is NP-complete already for bipartite and planar graphs.
%
%From a polyhedral point of view, many studies of packing polytopes associated to the Stable Set Problem and the Matching Problem have been proposed.
%
%The Stable Set Polytope is the convex hull of the incidence vectors of all stable sets and the Matching Polytope is the convex hull of the incidence vectors of all matchings.
%
%Many valid and facet-defining inequalities have been investigated for the Stable Set Polytope, see \cite{Letchford,Padberg1973,Chvatal,Oriolo,quasi-line,Rossi2001,Ventura,Rebennack2011,faenza2014solving,Yuri}.
%
%In particular, complete linear descriptions have been obtained for classes of graphs as line-graphs and quasi-line graphs, see \cite{quasi-line,Edmonds1965}.
%
The authors in \cite{Polyhedra} propose the first polyhedral study of the total matching problem, deriving facet-defining inequalities for the total matching polytope, which is defined as follows. 
Given a total matching $T$, the corresponding characteristic vector $\chi[T] \in \{0,1\}^D$ is defined as 
\begin{equation*}
  \chi[T]_a:=\left\{
  \begin{array}{@{}ll@{}}
    1 & \text{if}\ a \in T  \\
    0 & \text{if}\ a \in D\setminus T.
  \end{array}\right.
\end{equation*} 
In the following, we often write a characteristic vector of a total matching as $z = (x,y) \in \{0,1\}^{|V|} \times \{0,1\}^{|E|}$, with $x$ corresponding to the vertex variables and $y$ to the edge variables.
The \emph{total matching polytope} of a graph $G=(V,E)$ is defined as:
\[
P_{T}(G) := \mbox{conv}\{\chi[T]: T\mbox{ is a total matching of $G$} \}.
\]
%
%The first work on the Total Matching Problem appeared in \cite{TotalMatching}, where the authors derive lower and upper bounds on the size of a maximum total matching.
%

%

\smallskip 

\noindent {\bf Contributions and organization of the paper.}
The goal of this paper is to prove new results on the facial structure of $P_T(G)$, and highlight connections between $P_T(G)$ and the classical theory of polyhedral combinatorics. In Section~\ref{sec:preliminaries}, we introduce basic tools and employ them to show that the natural linear relaxation of $P_T(G)$ gives a complete description when $G$ is a tree. In Section~\ref{sec:facets}, we propose two new classes of inequalities for $P_T(G)$, dubbed \emph{balanced} and \emph{non-balanced lifted biclique} inequalities, and show that the former are always facet-defining, while the latter are facet-defining if $G$ is a bipartite graph. We also address complexity issues of the associated separation problem. In Section~\ref{sec:ef}, we give an extended formulation for $P_T(G)$ when $G$ is a complete bipartite graph. \revision{Extended formulations have been extensively used in polyhedral combinatorics to provide compact descriptions for polytopes that have exponentially many facets, see~\cite{SurveyExtended,kaibel2011extended} for surveys on classical results and, e.g.,~\cite{aprile2022extended,Aprile2021,del2023polynomial} for more recent developments.} Our extended formulation is based on Balas' technique to describe the convex hull of the union of polytopes~\cite{balas1998disjunctive} and on results on perfect graphs~\cite{chudnovsky2006strong,Chvatal}, \revision{and has a polynomial number of inequalities}. Using tools from~\cite{SurveyExtended} and through a careful analysis of the projection rays, we then project the extended formulation and show that $P_T(G)$ is described by the (exponential-sized, non-redundant) family of inequalities containing balanced and non-balanced lifted bicliques, and the inequalities from the natural linear relaxation of $P_T(G)$.% We conclude with open questions in Section~\ref{sec:conclusions} {\color{red}do we need the final section?}.

\smallskip 

\noindent {\bf Notation.} Given a graph $G=(V,E)$, let $n = |V|$ and $m = |E|$. 
For $v \in V$, we denote by $\delta(v)$ the set of edges incident to $v$. %and by $N(v)$ the set of vertices adjacent to $v$. 
%The degree of a vertex is $|\delta(v)|$.
%
 For a subset of vertices $U \subseteq V$, let $G[U]$ be the subgraph of $G$ induced by $U$.
%
%A total matching of maximum cardinality is denoted by $\nu_T(G)$.
%
%A total graph $T$ of a given graph $G$ is a graph with
%
%vertex set the vertices and edges of $G$ and two vertices are adjacent in $T$ 
%
%if and only if their corresponding elements are either adjacent or incident in $G$.
%
%A graph is \textit{chordal} if every cycle of length greater or equal than four has a \textit{chord}, that is, there is an edge connecting two non consecutive vertices of the cycle.
%
Given a graph $G$,  we let $V(G)$ (resp., $E(G)$) be its set of vertices (resp., edges). A \emph{biclique} $K_{r,s}$ is a complete bipartite graph, where the bipartition of $V(K_{r,s})$ is given by $(A,B)$ with $A=\{v_1,\dots,v_r\}$ and $B=\{w_1,\dots,w_s\}$. It is \emph{balanced} if $r=s$, \emph{non-balanced} otherwise.
%whose cardinality of the two sets of the vertex partition is $r$ and $s$, respectively. {For a bipartite graph, we denote the two sets of the bipartition by }

\section{Preliminaries}\label{sec:preliminaries}

We start with a natural linear relaxation of $P_T(G)$.
\begin{proposition}
Let $G(V,E)$ be a graph. The following is a linear relaxation for $P_T(G)$:
\begin{align}
\label{m6:c1} & x_v + \sum_{e \in \delta(v)} y_{e} \leq 1 & \forall v \in V \\
\label{m6:c2} & x_{v} + x_{w} + y_{e} \leq 1  & \forall e=\{v,w\} \in E \\
\label{m6:c3} & x_{v},y_{e} \geq 0 & \forall v \in V, \forall e \in E .
\end{align}
\end{proposition}
\eqref{m6:c1}-\eqref{m6:c3} are called, \emph{node}, \emph{edge}, and \emph{nonnegativity} inequalities, respectively.
In \cite{Polyhedra}, the authors prove that they are facet-defining for $P_T(G)$, for any graph $G$.
%
%{\color{red}move later Since $\nu_{T}(G)$ is polynomial for Trees, we study the linear description of such graphs.}
%
The following definition introduces a useful tool to study total matching problems.
\begin{definition}
Given a graph $G$, the \emph{total graph} $T(G)$ of $G$ is a graph whose
vertices are the elements of $G$, and where two vertices of $T(G)$ are adjacent 
if and only if their corresponding elements are adjacent in $G$.
\end{definition}
%
%\revision{

The \emph{stable set polytope} $\STAB(G)$ of a graph $G$ is the convex hull of the characteristic vectors of stable sets of $G$. The following fact has been observed in~\cite{Polyhedra}.

\begin{proposition}\label{Total}
Let $G$ be a graph and $T(G)$ its total graph. Then, $P_{T}(G) =\STAB(T(G))$.
\end{proposition}
%}
%

Proposition~\ref{Total} already allows us to characterize $P_T(G)$ when $G$ is a tree. 

\begin{theorem}\label{thm:tree}
Let $G$ be a tree. Then a complete and non-redundant description of $P_{T}(G)$ is given by
\eqref{m6:c1} -- \eqref{m6:c3}. %Hence, the optimization problem on $P_{T}(G)$ for a Tree graph can be solved in polynomial time.
%\begin{align*}
%& \sum_{e \in \delta(v)} y_{e} \leq 1-x_{v} & \forall v \in V \\
%& x_{v} + x_{w} \leq 1 - y_{e}  & \forall e=\{v,w\} \in E \\
%& y_{e},x_{v} \geq 0 & \forall v \in V, \forall e \in E.
%\end{align*}
\end{theorem}

\begin{proof}
Consider the total graph $T(G)$ of a tree $G$. 
In~\cite[Theorem 5]{Yannakakis}, it is shown that a connected graph is a tree if and only if its total graph is \emph{chordal}, i.e., for every cycle and every pair of non-consecutive vertices $u,v$ of the cycle, we have $\{u,v\} \in E(G)$. 
Hence, we have:
$$ P_T(G)=\STAB(T(G))=\bigg\{x \in \mathbb{R}^{|V(T(G))|}_{\geq 0} : \sum\limits_{v \in K}x_v \leq 1, \forall \hbox{ clique $K$ of } T(G) \bigg\},$$
where the first equality follows by Proposition~\ref{Total} and the second since chordal graphs are perfect (see, e.g.,~\cite{Schrijver2003}).
Let $K$ be a maximal clique in $T(G)$. 
Since $G$ has no cycles, the preimage of $K$ in $G$ is either a node and its neighborhood, or an edge and its endpoints.
Hence, a maximal clique inequality corresponds to an inequality of the type \eqref{m6:c1} - \eqref{m6:c2}. We deduce that~\eqref{m6:c1} - \eqref{m6:c3} give a complete description of $P_T(G)$. To observe that the description is non-redundant, recall that~\cite{Polyhedra} showed that all inequalites~\eqref{m6:c1} -- \eqref{m6:c3} define facets.\qed
\end{proof}

  %maximal cliques of the Stable Set Polytope of $T(G)$ correspond to basic inequalities \eqref{m6:c1} -- \eqref{m6:c3} of the initial graph.
%
%For the second, observe that \eqref{m6:c1} -- \eqref{m6:c3} contains $O(|V|+|E|)$ inequalities.

%

Theorem~\ref{thm:tree} gives an alternative, polyhedral proof of the fact that a maximum total matching in a tree can be found in polynomial time~\cite{Manlove}, and in fact shows that even a total matching of maximum weight (with weights defined over the elements of the tree) can be found in polynomial time. 

\section{New families of facet-defining inequalities}\label{sec:facets}

\subsection{Balanced biclique inequalities} 

The facet-defining inequality \eqref{m6:c2} can be seen as induced by a balanced biclique $K_{1,1}$.
We next derive a generalization of these inequalities, and  show that they are facet-defining for the total matching polytope of any graph. 
%
%As observed in~\cite{Leidner2012}, we have $\nu_{T}(K_{r,r})=r$.
%
%Then, we have a natural upper bound to obtain the following valid inequality.
%Given a graph $G$, a balanced biclique $K_{r,r}$ is a complete bipartite graph whose the cardinality of the two partition of vertices are the same.
%
\begin{lemma}\label{thm:balanced}
Let $G$ be a graph and $K_{r,r}$ be an induced balanced biclique of $G$ with $r \geq 2$. Then, the balanced biclique inequality:
\begin{equation}\label{complete}
 \sum\limits_{v \in V(K_{r,r})}{x_{v}}+\sum\limits_{e \in E(K_{r,r})}{y_{e}} \leq r
\end{equation}
is facet-defining for $P_{T}(G)$.
\end{lemma}
\begin{proof}
Let $V(K_{r,r})=A{\cup} B$. The validity of the inequality follows from $\nu_{T}(K_{r,r})=r$~\cite{Leidner2012}. Let $\Tilde{F} = \{z \in P_{T}(G) : \pi^{T}z = \pi_{0} \}$ be the face of $P_T(G)$ defined by~\eqref{complete}, and let $F= \{z \in P_{T}(G) : \lambda^{T}z = \lambda_{0} \}$ be a face of $P_T(G)$ such that $\Tilde{F} \subseteq F$. We prove that there exists $a\neq 0$ such that $(\lambda,\lambda_0) = a(\pi,\pi_0)$.

%
%
%the inequality defining $F$ is a scalar multiple of the biclique inequality, 
%so that $\Tilde{F} = F$ by maximality.
%
%First, the inequality \eqref{complete} is a valid inequality for $P_{T}(G)$, this result was proved in \cite{Leidner2012}.
%

Let $e=\{v,w\} \in E[K_{r,r}]$ with $v \in A$, $w \in B$, and define the total matchings $T_{v} := (A \setminus \{v\}) \cup \{e\}$ and $T_{w} := (B \setminus \{w\}) \cup \{e\}$.
Since $|T_v|=|T_w|=r$ and $\chi[A], \chi[B]\in \tilde F \subseteq F$, %the cardinality of the total matchings described is equal to the cardinality of the perfect matchings, this implies that 
we have $\chi[T_{v}],\chi[T_{w}] \in \Tilde{F}\subseteq F$.
Hence, $\lambda^{T} \chi[A] = \lambda^{T} \chi[T_{v}] = \lambda^{T} \chi[T_{w}]=\lambda_0.$ We deduce therefore that $\lambda_{v}=\lambda_{w}=\lambda_{e}$. Since $e \in E[K_{r,r}]$ was arbitrarily chosen, we deduce $ \lambda_v = \lambda_w = \lambda_{e}$ $\forall v \in A, w \in B, e \in E[K_{r,r}]$.%
%This holds $\forall u,v \in V$ and every $e \in E$. T
%thus $\lambda_{k}=a \in \mathbb{R}$ for all $k \in A_1 \cup A_2 \cup E[K_{r,r}]$. In particular, $\chi[A_1], \chi[A_2] \in F$. %the coefficients on $G[K_{r,r}]$. 
% 
%By fixing one of the total matchings $T$ introduced, we obtain also that $\lambda^{T}\chi[T]= \lambda_0 = a\pi_0$.
%

Now, consider an element $d \notin ( A \cup B \cup E[K_{r,r}])$ of $G$. Let $M$ be a perfect matching of $K_{r,r}$. Note that at least one of $T_1:=A \cup \{d\}$ and $T_2 := B \cup \{d\}$, and $T_3:= M \cup \{d\}$ is a total matching. We assume that $T_1$ is a total matching, the other cases following analogously. $\chi[A], \chi[T_1]\in \tilde F \subseteq F$ imply that $\lambda_d = 0$. %V(K_{r,r})$. Since $M_i$ does not cover $v$, $T:=M_i \cup \{v\}$ is a total matching of $G$. 
%$\chi[T_1]$ is a point of $P_T(G)$ that dominates $\chi[A]$ componentwise. Observe that $P_T(G)$ is a full-dimensional down-monotone polytope and $F$ is  not induced by a nonnegative inequality since $\lambda$ has at least two non-zero coefficients. Since $\chi[A] \in F$, we have therefore $\chi[T_1] \in F$.
%
%This implies that $\lambda_d = 0$.
%
%Then, consider $T_{A_i}^{e}:=A_i \cup \{e\}$, where $e \notin \delta(A_i)$, for $i=1,2$.
%
%It turns out that $T_{A_i}^{e}$ is a total matching and in particular its characteristic vector lies on $ \Tilde{F}$.
%
%Thus $\lambda_e = 0$, $\forall e \notin E(K_{r,r})$.
%
This fact completes the proof.
\qed
\end{proof}

\subsubsection{Separation of balanced biclique inequalities}

We next address the problem of separating balanced biclique inequalities of fixed cardinality. The following problem is NP-Complete~\cite{pandey2020maximum}.

\smallskip

\noindent\fbox{\parbox[c]{1\textwidth - 2\fboxsep - 2\fboxrule}{%
\noindent Name: Weighted Edge Biclique Decision Problem (\texttt{WEBDP}).

\noindent Input: A complete bipartite graph $G$ with edge weights $u \in {\mathbb{Z}^E}$, a number $k \in \mathbb{N}$. 

\noindent Decide: If there exists a subgraph of $G$ that is a biclique with vertex partition $(A,B)$ such that $\sum_{e \in A \times B} u(e) \geq k$.}}

\medskip

The NP-Completeness of \texttt{WEBDP} implies that the following problem is NP-Complete. 

\medskip

\noindent\fbox{\parbox[c]{1\textwidth - 2\fboxsep - 2\fboxrule}{%
\noindent Name: Weighted Edge Biclique Decision Problem, fixed Cardinality (\texttt{WEBDPC}).

\noindent Input: A complete bipartite graph $G'(V',E')$ with edge weights $u' \in \mathbb{N}^E$, numbers $k',q \in \mathbb{N}$ with $q \leq |V'|$.

\noindent Decide: If there exists a subgraph of $G'$ that is a biclique with vertex partition $(A,B)$, $|A|=|B|=q$, such that $\sum_{e \in A \times B} u(e) \geq k'$.
}}

\medskip

Indeed, \texttt{WEBDPC} is clearly in NP. {Suppose we want to solve \texttt{WEBDP} on input $G,u,k$, where $G$ has bipartition $(A_0,B_0)$. For $r \in \{1,\dots, |A_0|\}$, $\ell \in \{1,\dots, |B_0|\}$, consider the instance ${\cal I}(r,\ell)$ of \texttt{WEBDPC} on input $G',u',k',q$, defined as follows (we assume $r \geq \ell$, the other case being defined analogously). $G'$ is a complete bipartite graph, with bipartition $(A',B')$, where $A'=A_0$, $B'=B_0 \cup \bar B$, and $\bar B$ is a set of $r-\ell$ new vertices. For $e \in E(G)$, we let $u'_e=u_e + \|u\|_\infty$; for $e$ incident to $\bar B$, we let $u'_e = 1+ \max_{e \in E(G)} u'_e$. We moreover let $k'=k + r\ell \cdot \|u\|_\infty + r \cdot (r-\ell) \cdot (1+\max_{e \in E(G)}u'(e))$ and $q=r$. Observe that $u' \in \mathbb{N}^E$ and $k',q \in \mathbb{N}$.

Let $(A^*,B^*)$ be an optimal solution to the problem of selecting a biclique $(A,B)$ of $G'$ with $|A|=|B|=q$ maximizing $\sum_{e \in A \times B} u'(e)$. We claim that $\bar B \subseteq B^*$: if not, let $b \in \bar B\setminus B^*$, and $b' \in B^*\setminus \bar B$ (the former exists by hypothesis, the latter by $|\bar B|<r=|B^*|$). Since each edge incident to $b'$ has value at most $\max_{e \in E(G)}u'_e$, while each edge incident to $b$ has value $1+\max_{e \in E(G)}u'_e$, we contradict the optimality of $(A^*,B^*)$. Hence, ${\cal I}(r,\ell)$ is a yes instance for \texttt{WEBDPC} if and only if $$\begin{array}{lll} \sum_{e \in A^* \times B^*} u'(e) \geq k' & \Leftrightarrow & \sum_{e \in A^* \times (B^*\setminus \bar B)} (u(e) + \|u\|_\infty) + \sum_{e \in A^* \times \bar B} (1+ \max_{e \in E(G)}u'_e) \\ & & \quad \geq k + r\ell \cdot \|u\|_\infty + r \cdot (r-\ell) \cdot (1+\max_{e \in E(G)}u'(e)) \\ & \Leftrightarrow & \sum_{e \in A^* \times (B^*\setminus \bar B)} u(e) \geq k,\end{array}
$$
i.e., if and only if there exists a biclique $(\tilde A,\tilde B)$ of $G$ with $|\tilde A|=r, |\tilde B|=\ell$ such that $\sum_{e \in \tilde A \times \tilde B} u(e) \geq k$. Hence, the instance defined for $G,u,k$ is a yes-instance for \texttt{WEBDP} if and only if at least one of the instances ${\cal I}(r,\ell)$ is a yes-instance for \texttt{WEBDPC}, concluding the proof.}

Consider now the following problem.

\medskip

\noindent\fbox{\parbox[c]{1\textwidth - 2\fboxsep - 2\fboxrule}{%
\noindent Name: Separation Problem for Balanced Biclique Inequalities of a Given Size in 
 complete bipartite graphs (\texttt{SPBBIGS})

\noindent Input: A complete bipartite graph $G(A'\cup B',E)$, a point $(x^*,y^*) \in \mathbb{Q}^{|V|+|E|}_+$,  $r \in \mathbb{N}$, $r \leq |A'|$.

\noindent Decide: If there exists a violated balanced biclique inequality of $P_T(G)$ of $G$ with $r$ nodes on each side of the partition, that is, there exist $A\subseteq A', B \subseteq B', |A|=|B|=|r|$ such that:
\begin{equation}\label{eq:complete:separation}
 \sum\limits_{v \in A \cup B}{x^*_{v}}+\sum\limits_{e \in A \times B}{y^*_{e}} > r.
\end{equation}
}}

\smallskip

Clearly \texttt{SPBBIGS} belongs to NP. %Suppose we have an algorithm for \texttt{SPBBIGS}, we show an algorithm for \texttt{WEBDPC}. Let $G(A_0,B_0),u,k$ be an input to \texttt{WEBDP}. For $t \in \{1,\dots, |A_0|\}$, $\ell \in \{1,\dots, |B_0|\}$, we define and solve an instance ${\cal I}_{t,\ell}$ of \texttt{SPBBIGS} defined as follows. Assume without loss of generality that $t \geq \ell$; let $A'=A_0$, $B'=B_0 \cup |\bar B|$ with $|\bar B|=|A_0|-|B_0|$; let $x^*=0$, $y_e^*=t ||w||_\infty$ for $e$ incident to $B_0$, and $w_e=u_e $
Let $G,u,k,q$ be an input to \texttt{WEBDPC}, where $G$ has vertex bipartition $(A',B')$. Define $y^*_e=u_e \frac{q}{k-1}$ for each edge $e \in E(G)$, $x^*=0$, and consider the instance of \texttt{SPBBIGS} on input $G,(x^*,y^*)$, $r=q$. There exists an inequality of the form~\eqref{eq:complete:separation} separating $(x^*,y^*)$ if and only if for some $A\subseteq A'$, $B\subseteq B'$ with $|A|=|B|=r$, we have $$ \frac{q}{k-1}\sum_{e \in A \times B} u_e= \sum_{e \in A \times B} y^*_e =\sum_{v \in A \cup B} x^*_v +  \sum_{e \in A \times B} y^*_e  > r = q \, \Leftrightarrow \, \sum_{e \in A \times B} u_e > k -1 \Leftrightarrow \, \sum_{e \in A \times B} u_e \geq  k,$$
if and only if the instance of \texttt{WEBDPC} is a yes-instance. %Similarly, if no inequality of the form~\eqref{eq:complete:separation} is violated, we have $\sum_{e \in A \times B } u_e < k$, for all $A\subseteq A', B \subseteq B'$, $|A|=|B|=r=q$.

\subsection{Non-balanced bicliques inequalities} 

Consider a non-balanced biclique
$K_{r,s}$, with $s>r\geq 2$, of a graph $G$.  
By mimicking~\eqref{complete}, it is natural to ask whether \begin{align}\label{non-balanc-bicl}
    \sum\limits_{v \in V(K_{r,s})}x_v + \sum\limits_{e \in E(K_{r,s})}y_e \leq s 
\end{align} defines a facet. This inequality is indeed valid but not facet-defining. We next give a strengthening of~\eqref{non-balanc-bicl} and show that it defines a facet when $G$ is a bipartite graph. %This strengthening can be obtained starting from $\sum\limits_{v \in S}x_v + \sum\limits_{e \in E(K_{r,s})}y_e \leq s$ and iteratively applying sequential lifting~\cite{Zemel,Padberg1973}. 

{
\begin{proposition}\label{lifting biclique}
Let $K_{r,s}$ with $s>r\geq 2$ be a non-balanced biclique {with vertex partition $(A,B)$. Let $A_1\subseteq A$, $B_1\subseteq B$ with $|A|>|A_1|>|B_1|$ such that, if $|B_1|=0$, then $|A_1|=1$. Define $\beta:=\frac{s+|A_1|-r-|B_1|}{|A_1|-|B_1|}$.}
Then, the \textit{non-balanced lifted biclique inequality}:
\begin{align}\label{biclique}
    \sum\limits_{v \in A} \alpha_v x_{v}+ \sum\limits_{w \in B} {\alpha_w}x_{w}+ \sum\limits_{e \in E(K_{r,s})} {\alpha_e} y_e \leq {|A_1|(\beta-1)+|A|},& & 
\end{align}
where { for each element $d$ of $K_{r,s}$ we let}
\begin{equation*}
  \alpha_{d}=\left\{
  \begin{array}{@{}ll@{}}
    {\beta} & \text{if}\ { d \in A_1 \cup B_1 \cup (A_1 \times B_1)}, \\
    1 & \text{otherwise},
  \end{array}\right.
\end{equation*} 
is facet-defining for $P_T(K_{r,s})$.
\end{proposition}

\begin{proof}
$s>r$ implies that $s+|A_1|-r-|B_1|>|A_1|-|B_1|$. Because of $|A_1|>|B_1|$, we then deduce $\beta >1$. Let $A_2:=A\setminus A_1$, $B_2:=B\setminus B_1$. We deduce $s+|A_1|-r-|B_1|=|B_2|-|A_2|$, $\beta=\frac{|B_2|-|A_2|}{|A_1|-|B_1|}$, and $\beta|A_1| + |A_2|=\beta |B_1|+|B_2|$.

We first show that~\eqref{biclique} is valid for $P_T(K_{r,s})$. Let $T$ be an inclusionwise maximal total matching of $K_{r,s}$ and $z$ its characteristic vector. Assume first $T\cap A \neq \emptyset$. Then $T \cap B = \emptyset$. For each $v \in A$, $T$ contains exactly one of $v$ and an edge incident to $v$. Since all edges incident to $v \in A_2$ have the same coefficients in~\eqref{biclique} as $v$, we can assume w.l.o.g.~that $A_2\subseteq T$. Then the left-hand side of~\eqref{biclique} is maximized when $T = A_2 \cup A_1$ and is therefore at most $|A_2|+\beta|A_1|=|A|-|A_1|+\beta|A_1|=|A_1|(\beta-1)+|A|$. We deduce that $z$ satisfies~\eqref{biclique}.

Next, assume $T\cap B \neq \emptyset$. Then $T \cap A = \emptyset$. Similarly to the above, we conclude that w.l.o.g.~$T=B_1\cup B_2$, hence the left-hand side of~\eqref{biclique} computed in $z$ is at most $|B_2|+\beta |B_1|=\beta|A_1| + |A_2|$ and again we deduce that $z$ satisfies~\eqref{biclique}. 

Last, assume $T\cap A, T \cap B=\emptyset$. Then $T$ is a matching of $K_{r,s}$. $T$ contains at most $\min\{|A_1|,|B_1|\}=|B_1|$ edges with coefficient $\beta$ and $\min\{|A|,|B|\}$ edges in total. Hence, the left-hand side of~\eqref{biclique} computed in $z$ is at most $\beta |B_1| + |B|-|B_1|= \beta|B_1| + |B_2| = \beta|A_1| + |A_2|$,
concluding the proof that~\eqref{biclique} is valid for $P_T(K_{r,s})$.

We next show that the face $\tilde F$ of $P_T(G)$ defined by~\eqref{biclique} is a facet of $P_T(G)$. Let $F= \{z \in P_{T}(G) : \lambda^{T}z = \lambda_{0} \}$ be a face of $P_T(G)$ such that $\tilde{F} \subseteq F$. First observe that $\chi[A], \chi[B] \in \tilde F \subseteq F$. Moreover, by hypothesis, $|B_2|>|A_2|\geq 1$. Hence, we can repeat the same argument as in the proof of Lemma~\ref{thm:balanced} and deduce that $\lambda_v=\lambda_w=\lambda_e=:\lambda_2$ for all $v \in A_2, w \in B_2$ and $e \in E(K_{r,s})$ with at least one endpoint in $A_2 \cup B_2$. Again by hypothesis, $|A_1|\geq 1$. If $|B_1|\geq 1$, then we deduce $\lambda_v=\lambda_w=\lambda_e=:\lambda_1$ for all $v \in A_1, w \in B_1$ and $e \in A_1 \times B_1$. Moreover, $\tilde F \subseteq F$ implies $\lambda^T\chi[A]= \lambda^T\chi[B]$, which in turn implies $\lambda_2 |A_2| + \lambda_1 |A_1|= \lambda_2|B_2|+\lambda_1|B_1|$ or equivalently, $\lambda_1=\lambda_2 \frac{|B_2|-|A_2|}{|A_1|-|B_1|}=\lambda_2 \cdot \beta$, as required.  If $|B_1|=0$, then by hypothesis $A_1=\{v_1\}$, and using again $\lambda^T\chi[A]= \lambda^T\chi[B]$, we deduce $\lambda_{v_1} + (|A|-1) \lambda_2 = \lambda_2 |B|$, or equivalently $\lambda_{v_1}=\lambda_2 (|B|-|A|-1)=\lambda_2 \frac{|B_2|-|A_2|}{|A_1|-|B_1|}=\lambda_2 \cdot \beta$, as required.\qed
\end{proof}
}

The following proposition shows that, in a bipartite graph $G$, the inequalities from Proposition~\ref{lifting biclique} define facets.

\begin{lemma}\label{thm:non-balanced}
Let $G(V,E)$ be a bipartite graph and $K_{r,s}$ a subgraph of $G$. The non-balanced lifted biclique inequalities \eqref{biclique}
are facet-defining for $P_{T}(G)$.
\end{lemma}

\begin{proof}
Let $V(K_{r,s})= A \cup B$ and
$F$ be the face induced by a non-balanced lifted biclique inequality associated to $K_{r,s}$.
%
%We have to find $n+m$ affinely independent points belonging to $F$. 
%
By Proposition~\ref{lifting biclique}, we have a set ${\cal S}$ of $|V(K_{r,s})|+|E(K_{r,s})|$ affinely independent points that lie in $F$ whose support is contained in the elements of $K_{r,s}$. For each element $d$ of $G$ with $d \notin Q:=A \cup B \cup \revision{E[K_{r,s}]}$, we give a total matching $M_d$ such that $\chi[M_d] \in F$, and ${\cal S} \cup \{M_d\}_{d \in (V \cup E)\setminus Q}$ is linearly independent. 

Let $d$ be an element of $G$ with $d \notin Q$. If $d$ is not adjacent to $A$, let $M_d = A \cup \{d\}$. Else, $d$ is not adjacent to $B$ since $G$ is bipartite, and we let $M_d = B \cup \{d\}$. Clearly, $M_d \in F$, and the matrix having as columns vectors from ${\cal S} \cup_{d \in (V \cup E) \setminus Q} M_d$ has the following form:
\begin{center}
$M=\left[
 \begin{array}{c|c}
M_1 & M_2  \\ \hline
\mathbf{0} & {\chi[\{d\}]}_{d \in (V\cup E)\setminus Q}
\end{array}\right],
$
\end{center}
where the first set of rows is indexed over elements of $K_{r,s}$, $M_1$ is the collection of vectors from ${\cal S}$ restricted to nodes of $K_{r,s}$, and $M_2$ is an appropriate matrix. Since $M_1$ has full rank and the bottom right submatrix of $M$ is the identity matrix, $M$ has full rank, and the thesis follows.
\qed
\end{proof}
\section{The Total Matching Polytope of Complete Bipartite Graphs}\label{sec:ef}

In this section, we give a complete and non-redundant description of $P_T(G)$ when $G$ is a complete bipartite graph. Our argument is as follows. In Section~\ref{sec:ef:algo}, we give a (simple) algorithm for solving the maximum weighted total matching problem on a complete bipartite graph $G$. In Section~\ref{sec:ef:ef}, we use this algorithm, results on perfect graphs, and Balas' classical theorem on the convex hull of the union of polytopes to give a compact extended formulation $Q$ for $P_T(G)$. Then, in Section~\ref{sec:ef:proj}, we study the projection cone associated to $Q$ to deduce the following.

\begin{theorem}\label{thm:complete-description}
Let $G$ be a complete bipartite graph. A complete and non-redundant description of $P_{T}(G)$ is given by the basic inequalities~\eqref{m6:c1} -- \eqref{m6:c3} and, for each balanced (resp., non-balanced) complete bipartite subgraph of $G$, the balanced biclique inequality~\eqref{complete} (resp., non-balanced lifted biclique inequality~\eqref{biclique}).
\end{theorem}

Throughout the section, fix a complete bipartite graph $G=K_{r,s}$, with, as usual, $V(G):=V=A\cup B$, $A=\{v_1,\dots,v_r\}$, and $B=\{w_1,\dots,w_s\}$. 

\subsection{Algorithm}\label{sec:ef:algo}
A total matching $T$ of $K_{r,s}$ satisfies at least one of $T \cap A = \emptyset$ and
$T \cap B = \emptyset$. For $U \in\{A,B\}$, let $T(K_{r,s})\setminus U$ be the subgraph of the total graph $T(K_{r,s})$ of $K_{r,s}$ obtained by removing nodes corresponding to elements of $U$ and the edges incident to them.  

By Proposition~\ref{Total}, we can solve the maximum weighted total matching problem on $G$ by solving the maximum weighted stable set on $T(K_{r,s})\setminus U$ for $U \in\{R,S\}$, and selecting the solution of maximum weight. The next lemma shows that such graphs have a special structure.

%We recall that $T(K_{r,s} \setminus U)$ is formed by a list of cliques $K_{v_1},K_{v_2} \dots K_{v_r}$ for each node $v_$ and $K_{w_1},K_{w_2}, \dots K_{w_s}$.
\begin{lemma}\label{lem:perfect-graphs}
Let $U \in\{A,B\}$. The graph $T(K_{r,s})\setminus U$ is perfect.
\end{lemma}
\begin{proof}
Suppose w.l.o.g.~that $U = B$.
We denote by $q(i,j)$ (resp., $v'_i$), the vertex or $T(K_{r,s})$ associated to the edge $e=\{v_{i},w_{j}\}$ (resp., vertex $v_i \in A$) of the original graph $K_{r,s}$.
We prove that neither $T(K_{r,s})\setminus B$ nor $\overline{T(K_{r,s} )\setminus B}$ contain an odd \revision{hole} with $5$ or more nodes. The statement then follows from the well-known characterization of perfect graphs~\cite{chudnovsky2006strong}. 

We start with  $T(K_{r,s})\setminus B$. By construction, for $i=1,\dots, r$, every vertex $v'_{i}$ lies in exactly one inclusionwise maximal clique (corresponding to edges adjacent to $v_i$ in $K_{r,s}$) and it is not adjacent to any node outside this clique. Thus, no odd \revision{hole} with at least $5$ nodes contains a vertex $v'_{i}$.
Hence, an odd \revision{hole} $C$ contains only vertices of the kind $q(i,j)$. We call $i$ (resp.~$j$) the \emph{first} (resp.~\emph{second}) \emph{entry} of the vertex $q(i,j)$. Note that no three consecutive vertices of $C$ can share the same first or second entry; on the other hand, two consecutive vertices of $C$ must share the first or the second entry. Hence, if we let 
$C=\{q_0,q_1,\dots,q_{k-1}\}$, we can assume w.l.o.g.~that, for $\ell$ odd, $q_\ell$ shares the first entry with $q_{\ell+1}$ and the second entry with $q_{\ell-1}$ (indices are taken modulo $k$). However, this contradicts $k$ being odd. 

We now focus on $\overline{T(K_{r,s}) \setminus B}$. Let $C=\{q_0,\dots, q_{k-1}\}$, $k \geq 5$ be an odd \revision{hole} in $\overline{T(K_{r,s}) \setminus B}$. First observe that $V(C)\cap A=\emptyset$. Indeed, suppose by contradiction that $v'_i \in V(C) \cap A$, and let w.l.o.g.~$v'_i=q_0$. Then $q_{\lceil\frac{k}{2}\rceil}=q(i,j)$ and $q_{\lfloor\frac{k}{2}\rfloor}=q(i,\ell)$ for some indices $j,\ell$. Then $q_{\lceil\frac{k}{2}\rceil}$ and $q_{\lfloor\frac{k}{2}\rfloor}$ are not adjacent in $\overline{T(K_{r,s}) \setminus B}$, a contradiction. Hence, $V(C) \cap A=\emptyset$, and let w.l.o.g.~$q_0=q(1,1)$. %We distinguish two cases.

%First assume that $k= 5$. Since $q_2,q_3$ are not adjacent to $q_0$ but they are adjacent to each other, we can assume w.l.o.g.~that $q_2=q(1,2)$, $q_3=q(2,1)$. Since $q_1$ is adjacent to $q_0$ but not to $q_3$, we must have $q_1=q(2,t)$, with $t \neq 1$. Since $q_1$ is adjacent to $q_2$, $t \neq 2$. Symmetrically, $q_4=q(p,2)$ with $p \neq 1,2$. On the other hand $q_1$ and $q_4$ are not adjacent, hence they must share one of their two entries. Hence either $t=2$ or $p=2$, a contradiction.

\revision{When $k=5$, $C$ is also an odd hole in $T(K_{r,s}) \setminus B$, which as argued above cannot exist.} Hence assume $k \geq 7$. We can assume w.l.o.g.~that~$q_{\lfloor\frac{k}{2}\rfloor}=q(1,2)$, $q_{\lceil\frac{k}{2}\rceil}=q(2,1)$. Since $q_{\lfloor\frac{k}{2}\rfloor-1}$ is not adjacent to $q_0$ or $q_{\lceil\frac{k}{2}\rceil}$, we must have $q_{\lfloor\frac{k}{2}\rfloor-1}=q(t,1)$ for $t\neq 1,2$. Symmetrically, $q_{\lceil\frac{k}{2}\rceil+1}=q(1,p)$ for $p\neq 1,2$. Again, using the fact that $q_{\lfloor\frac{k}{2}\rfloor-1}$ and $q_{\lceil\frac{k}{2}\rceil+1}$ are not adjacent, we deduce $t=1$ or $p=1$, a contradiction. \qed
\end{proof}

Note that, if we consider $T(K_{r,s})$ instead of $T(K_{r,s}) \setminus U$ for $U \in \{A,B\}$, the graph is no longer perfect. For instance, it can be easily checked that the total graph $T(K_{2,2})$ contains an odd hole. 

Lemma~\ref{lem:perfect-graphs} allows us to use classical semidefinite techniques~\cite{SeparationOptimization} to solve the maximum weighted stable set problem on $T(K_{r,s})\setminus U$ for $U \in \{A,B\}$. However, in our case we do not need to employ semidefinite programming, because of the following observation.

\begin{obs}\label{obs:linear-cliques}
Let $U \in \{A,B\}$. The cliques of $T(K_{r,s})\setminus U$ correspond in $G$ either to a node in $\{A,B\}\setminus \{U\}$ and the edges incident to it, or to edges incident to a node in $U$. In particular, $T(K_{r,s})\setminus U$ has $O(r+s)$ maximal cliques. \end{obs}

\subsection{Extended formulation}\label{sec:ef:ef}

%We can apply a classical result of Balas~\cite{balas1998disjunctive} to obtain an extended formulation for $P:=P_T(K_{r,s})$. 
Define the two polytopes
$$P_A:=\{ z \in P_T(K_{r,s}) : z_w = 0 \hbox{ for $w \in B$}\}, \quad P_B:=\{ z \in P_T(K_{r,s}) : z_v = 0 \hbox{ for $v \in A$}\}.$$
Using Lemma~\ref{lem:perfect-graphs} and the description of the stable set polytope of perfect graphs~\cite{Chvatal}, we can describe $P_A$ (and analogously $P_B$) completely using cliques inequalities:
$$P_A=\{ z \geq  0 : \sum\limits_{u \in K}z_u \leq 1 \hbox{ for $K$ clique of }V(T(K_{r,s})\setminus B)\},$$
and Observation~\ref{obs:linear-cliques} implies that this description has linear size. %allows for a linear-sized descr description of $P_R$ (and symmetrically, of $P_S$). 
We deduce the following.

\begin{corollary}\label{cor:PR-PS}
$$P_A= \bigg\{(x,y) \in \mathbb{R}_{\geq 0}^{|A|+|E|}: \qquad x_v + \sum_{e \in \delta(v)}y_e \leq 1, \quad  \forall v \in A; \qquad \sum_{e \in \delta(w)}y_e \leq 1, \quad \forall w \in B \bigg\},$$
$$P_B= \bigg\{(x,y) \in \mathbb{R}_{\geq 0}^{|B|+|E|}: \qquad x_w + \sum_{e \in \delta(w)}y_e \leq 1, \quad \forall w \in B; \qquad \sum_{e \in \delta(v)}y_e \leq 1, \quad \forall v \in A \bigg\}.$$
\end{corollary}

Following the discussion from Section~\ref{sec:ef:algo}, we can write $P:=P_T(G) = conv(P_A \cup P_B)$. Balas showed that the convex hull of the union of two polytopes has an extended formulation that can be easily described in terms of the original formulations of the polytopes~\cite{balas1998disjunctive}.
When applied to $P_A$, $P_B$ defined as above, %Theorem~\ref{thm:balas} 
Balas' result gives the extended formulation for $P$ from Corollary~\ref{cor:UnionPolytope},  where, for later usage, we also report certain dual multipliers.% Note that we used equation below we used the fact that node variables $x_v$ for $v \in A$ (resp.~$x_w$ for $w \in B$) do not appear in $P_B$ (resp.~$P_A$), allowing us to reduce the number of variables. Edge variables $y_e$ for $e \in E$ appear in both $P_A$ and $P_B$ (as $y^1$ and $y^2$, respectively), but we use the equality $y_e = y^1_e + y^2_e$ to project out $y^2_e$ for $e \in E$. Similarly, we use $\lambda_1 + \lambda_2 = 1$ to project out $\lambda_2$.

\begin{corollary}\label{cor:UnionPolytope} The following is an extended formulation for $P$:
\begin{align*}%\label{UnionPolytope}
    Q:= \bigg\{(x,y,\lambda_1,y_{e}^{1}) \in \mathbb{R}^{|V|+|E|+1+|E|}: & \qquad
    x_v + \sum_{e \in \delta(v)}y_{e}^{1} - \lambda_1\leq 0 , \forall v \in A & \qquad [u^1_v]\\
    \sum_{e \in \delta(w)}y_{e}^{1} - \lambda_1 \leq 0,  \forall w \in B \qquad [u^1_w] & \qquad
    x_w + \sum_{e \in \delta(w)}(y_e - y_{e}^{1})+ \lambda_1\leq 1,   \forall w \in B & \qquad [u^2_w]\\
    \sum_{e \in \delta(v)}(y_e - y_{e}^{1}) + \lambda_1\leq 1, \forall v \in A  \qquad [u^2_v] &  \qquad
%    x_v \leq \lambda_1, &  \forall v \in R \qquad [u^{\lambda_1}_v]\\
%    -x_v \leq 0,& \forall v \in R  \\ 
%    x_w \leq \lambda_2, &  \forall w \in S \qquad [u^{\lambda_2}_w]\\
%    - x_w \leq 0, &\forall w \in S \\  
%    y_{e}^{1} \leq \lambda_1, &  \forall e \in E \qquad [u^{1}_e]\\
    - y_{e}^1 \leq 0,  \forall e \in E \quad [{u^{1}_e}] & -x_v \leq 0, \ \forall v \in V\\
%    y_{e}-y_{e}^{1} \leq \lambda_2, & \forall e \in E \qquad [u^{2}_e]\\
    - y_{e} + y_{e}^{1} \leq 0,  \forall e \in E \qquad [u^{2 }_e]   
%    & \lambda_1 + \lambda_2 = 1,  \quad [u^{\lambda}]\\
   % x_v,y_{e}^{1} \geq 0, \forall v \in V,\forall e \in E\\
    & \qquad -\lambda_1 \leq 0, \quad [{u^{\lambda_1}}]  & \lambda_1 \leq 1  \quad [{u^{\lambda_2}}]
    \quad \bigg\}.
\end{align*}
\end{corollary}
\subsection{Projection}\label{sec:ef:proj}

In order to project the extended formulation defined in the previous section to the space $P$ lives in, we study the associated projection cone. 
\begin{theorem}~\cite[Theorem 2.1]{SurveyExtended}\label{thm:projection-cone}
Let $Q=\{(x,z) \in \mathbb{R}^{n}_{\geq 0} \times \mathbb{R}^{p} : Ax + Bz \leq d\}$ where $A,B$ have $m$ rows, and define its projection cone $C_P:=\{ u \in \mathbb{R}^{m} : uB = 0, u \geq 0 \}$. The projection of $Q$ onto the $x$-space is
$    \proj_{x}(Q)=\{ x \in \R^{n}_{\geq 0} : uAx \leq ud, \hbox{ for all extreme rays $u$ of $C_P$}\}.$
\end{theorem}

We next deduce a description of $P$ in terms of the projection cone from Theorem~\ref{thm:projection-cone}. \begin{lemma}\label{cone:rays}
$$P=\{ (x,y) \in \R^{n+m}_{\geq 0} : \sum\limits_{v \in A}u^{1}_vx_v + \sum\limits_{w \in B}u^{2}_wx_w + \sum\limits_{e=\{v,w\} \in E}\min_{j=1,2}(u^{j}_v + u^{j}_w)y_e \leq \max_{j=1,2}\sum\limits_{w \in V}u^{j}_w, \forall u \in Y \},$$
where $Y$ is the set of vectors $u \in \R_{\geq 0}^{2(|A|+|B|)}$ that satisfy $2(|A|+|B|)-1$ linearly independent constraints from the set 
\begin{align} 
u^{j}_v & = 0 & \hbox{ for $v \in V, {j \in \{1,2\}}$}  \nonumber \\
u_v^1 + u_w^1 & =  u_v^2 + u_w^2 \quad \hbox{ for $v \in A, w \in B$.} \label{eq:ciao} \\ 
\sum_{v \in V}u^1_v & =  \sum_{v \in V} u^2_v. \nonumber
\end{align}
\end{lemma}
\begin{proof} 
By specializing the description of $C_P$ from Theorem~\ref{thm:projection-cone} to the submatrix of the constraint matrix in Corollary~\ref{cor:UnionPolytope} corresponding to variables to be projected out, we obtain: 
\begin{align}
\label{c1new}   C_P= \bigg\{u {= (\{u^j_d\}_{j=1,2; d \in A \cup B \cup E}, u^{\lambda_1},u^{\lambda_2}) \in \mathbb{R}^{2(|V|+|E|+1)}}: & \nonumber \\ u^{1}_{v} + u^{1}_w -u^{1}_e =u^{2}_{v}+u^{2}_{w}-u^{2}_e , &\  \forall e=\{v,w\} \in E \\
\label{c2new}   \sum_{v \in V}u_{v}^{1}  + u^{\lambda_1} =  \sum_{v \in V}u_{v}^{2} + u^{\lambda_2}, & \\
%\label{c2bisnew} \sum_{v \in V}u_{v}^{2} - u^{\lambda_2} = 0 &\\
\label{c3new}    u \geq 0 \bigg\}.
\end{align}

Using Theorem~\ref{thm:projection-cone}, there is a complete description of $P$ containing only inequalities of the following form:
\begin{align}
\label{ineq} \sum\limits_{v \in A}u^{1}_vx_v + \sum\limits_{w \in B}u^{2}_wx_w + \sum\limits_{e=\{v,w\} \in E}(u^{2}_v + u^{2}_w - u^{2}_e)y_e \leq \sum\limits_{w \in V}u^{2}_w + u^{\lambda_2}
\end{align}
where $u$ is an extreme ray of $C_P$. %Note that by~\eqref{c1new}, we have $u^{2}_v + u^{2}_w - u^{2}_e= u^{1}_v + u^{1}_w - u^{1}_e$ for $e \in E$ and by~\eqref{c2new} we have $\sum\limits_{w \in V}u^{2}_w + u^{\lambda_2}=\sum\limits_{w \in V}u^{1}_w + u^{\lambda_1}$. 

We next claim that, in~\eqref{ineq}, we can assume w.l.o.g.~that, for $e \in E$, at least one of $u^{1}_e, u^{2}_e$ is equal to $0$. Indeed, since those variables are nonnegative, if they are both strictly positive we can decrease them by $\min\{u^{1}_e, u^{2}_e\}>0$ and obtain a stronger inequality~\eqref{ineq}. Similarly, at least one of $u^{\lambda_1}=0$, $u^{\lambda_2}=0$ holds. Hence, using~\eqref{c1new} and~\eqref{c2new}, we can rewrite~\eqref{ineq} as 
\begin{align*}
\label{ineq-rewrite} \sum\limits_{v \in A}u^{1}_vx_v + \sum\limits_{w \in B}u^{2}_wx_w + \sum\limits_{e=\{v,w\} \in E}\min_{j=1,2}(u^{j}_v + u^{j}_w)y_e \leq \max_{j=1,2}\sum\limits_{w \in V}u^{j}_w.
\end{align*}

\iffalse 
From the equation, we can write the previous inequalities in the following equivalent form:
\begin{align}
\label{ineq1} \sum\limits_{v \in A}u^{1}_vx_v + \sum\limits_{w \in B}u^{2}_wx_w + \sum\limits_{e=\{v,w\} \in E}(u^{1}_v + u^{1}_w - u^{1}_e)y_e \leq \sum\limits_{v \in V}u^{1}_v + u^{\lambda_1}
\end{align}
\fi

Let $\overline u$ be an extreme ray of $C_P$. \revision{Recall that $\overline{u} \in \mathbb{R}^{2(|V|+|E|+1)}$}. We first claim that the vector obtained from $\overline u$ by projecting out $\{u_e^{1},u_e^{2}\}_{e \in E}, u^{\lambda_1}, u^{\lambda_2}$ is a nonnegative vector that satisfies $2(|A|+|B|)-1$ linearly independent constraints from~\eqref{eq:ciao}. By construction, $\overline u$ satisfies at equality a set ${\cal S}$ of $2(|A|+|B|)+2|E|+1$ linearly independent constraints from~\eqref{c1new}--\eqref{c3new}. By basic linear algebra, any set of linearly independent constraints from~\eqref{c1new}--\eqref{c3new} tight at $\overline u$ can be enlarged to a linearly independent set of inequalities tight at $\overline u$ of maximum cardinality. Hence we can assume w.l.o.g.~that ${\cal S}$ contains the following set ${\cal S}'$ of linearly independent constraints. For $e \in E$, if $\overline u_{e}^{2}=\overline u_{e}^{1}=0$, then constraints $u_{e}^{2}=0$, $u_{e}^{1}=0$ belong to ${\cal S}'$. Else, from what argued above, we have $\overline u_{e}^{j}>0$ and $\overline u_{e}^{3-j}=0$ for some $j \in \{1,2\}$, and we let ${\cal S}'$ contain \revision{$u_{e}^{3-j}=0$} and $u^{1}_{v} + u^{1}_w -u^{1}_e =u^{2}_{v}+u^{2}_{w}-u^{2}_e$. Similarly, either $u^{\lambda_1}=0$ and $u^{\lambda_2}=0$ are both valid and we let them belong to ${\cal S}'$, or we let the one of them that is valid and $\sum_{v \in V}u_{v}^{1}  + u^{\lambda_1} =  \sum_{v \in V}u_{v}^{2} + u^{\lambda_2}$ belong to ${\cal S}'$. It is easy to see that constraints in ${\cal S}'$ are linearly independent, hence ${\cal S}\setminus {\cal S}'$ is a set of $2(|A|+|B|)-1$ linearly independent constraints. 

Note that an inequality~\eqref{c1new} belongs to ${\cal S}\setminus {\cal S}'$ only if both the variables $u_e^{1}$ and $u_e^{2}$ appearing in its support are set to $0$. In particular, only if $\overline u$ satisfies $\overline u^{1}_{v} + \overline u^{1}_w =\overline u^{2}_{v}+\overline u^{2}_{w}$. Similarly, a constraint~\eqref{c2new} belongs to ${\cal S}\setminus {\cal S}'$ only if $\sum_{v \in V}\overline u_{v}^{1} =  \sum_{v \in V}\overline u_{v}^{2}$. Hence, $\overline u$ satisfies at equality $2(|A|+|B|)-1$ linearly independent constraints from~\eqref{eq:ciao}, and the claim follows.

Conversely, any nonnegative vector in the components $\{u_v^1,u_v^2\}_{v \in V}$ that satisfies any set of constraints from~\eqref{eq:ciao} can be extended to a vector of $C_P$ by appropriately adding components $u^{1}_e, u^{2}_e, u^{\lambda_1}, u^{\lambda_2}$, concluding the proof. 
\qed
\end{proof}

%that Consider the projection cone associated with the polytope $P$.
%
%Notice that $\lambda_2 = 1 - \lambda_1$, thus our system of inequalities can be further reduced by removing the multiplier $u^{\lambda}$.
%
%Then, the projection cone can be written in the following form:
%
%\begin{align}
%\label{c1new}   C_P:= \bigg\{u: u^{1}_{v} + u^{1}_w -u^{1}_e =u^{2}_{v}+u^{2}_{w}-u^{2}_e , & \forall e=\{v,w\} \in E \\
%\label{c2new}   \sum_{v \in V}u_{v}^{1}  + u^{\lambda_1} =  \sum_{v \in V}u_{v}^{2} + u^{\lambda_2} & \\
%\label{c2bisnew} \sum_{v \in V}u_{v}^{2} - u^{\lambda_2} = 0 &\\
%\label{c3new}    u \geq 0 
%\bigg\}
%\end{align}
%

We call \emph{fundamental} a vector $u \in Y$ whose associated inequality (as in Lemma~\ref{cone:rays}) defines a facet of $P$ different from~\eqref{m6:c1}--\eqref{m6:c3},~\eqref{complete},~\eqref{biclique}. %, and {\color{red}there is no valid $u'$ whose associated inequality is equivalent up to nonnegative scaling to the one associated to $u$, but the support of $u'$ is strictly contained in the support of $u$. }We call a vector $u$ that belongs to the set $Y$ defined in Lemma~\ref{cone:rays} \emph{valid}, and an inequality that is obtained in the description of $P$ given in Lemma~\ref{cone:rays} from a valid $u$ \emph{legal}. We moreover call a valid $u$ \emph{fundamental} if the following properties hold: \begin{itemize}\item all its components are integer numbers with gcd $1$;\item the inequality associated to $u$ (as in Lemma~\ref{cone:rays}) defines a facet of $P$ different from~\eqref{m6:c1}--\eqref{m6:c3}, and \item there is no valid $u'$ whose associated inequality is equivalent up to nonnegative scaling to the one associated to $u$, but the support of $u'$ is strictly contained in the support of $u$. \end{itemize} We say that a set of $2(|A|+|B|)-1$ linearly independent constraints from~\eqref{eq:ciao} \emph{support} a valid $u$ if they are tight at $u$. 
A set ${\cal S}$ of $2(|A|+|B|)-1$ linearly independent constraints from~\eqref{eq:ciao} is said to \emph{support} $u$ if $u$ satisfies all constraints in ${\cal S}$. For a set ${\cal S}$ supporting a vector $u \in Y$, we let $G({\cal S})$ be the graph that contains all vertices of $K_{r,s}$, colors a vertex $v$ \emph{blue} (resp.~\emph{red}) if $u_v^1=0$ (resp., $u_v^2=0$) belongs to ${\cal S}$, and contains edge $vw$ if $u_v^1 + u_w^1 =  u_v^2 + u_w^2$ belongs to ${\cal S}$. Note that a node can be colored both blue and red in $G({\cal S})$ - we call such nodes \emph{bicolored}. A node that is colored with exactly one of red and blue is \emph{monochromatic}. A connected component of $G({\cal S})$ is \emph{non-trivial} if it contains at least two nodes. 

Let ${\cal S}$ be the set of constraints supporting a fundamental vector $u$. %and let $ax\leq b$ be the inequality associated to $u$. 
${\cal S}$ is called \emph{canonical} if, among all sets supporting $u$, $G({\cal S})$ maximizes the number of colored nodes (with each bicolored node counting twice) and, subject to the previous condition, maximizes the number of edges. %Notice that, to obtain a valid description of $P$, it is enough to describe canonical sets of inequalities associated to minimal, valid vectors, and add to those inequalities~\eqref{m6:c1}--\eqref{m6:c3}.

\begin{lemma}\label{lem:cycle} Let $u$ be fundamental, $\cal S$ a canonical set supporting it, and ${\cal I}$ the isolated nodes of $G({\cal S})$.  
\begin{enumerate} 

\item For each edge $\{v,w\}$ of $G({\cal S})$, $v$ and $w$ are monochromatic and colored with opposite colors.

\item If $\sum_{v \in V}u^1_v =  \sum_{v \in V} u^2_v$ does not belong to ${\cal S}$, there is exactly one non-trivial connected component $C$ in $G({\cal S})$, all nodes of $C$ are monochromatic, and all nodes from ${\cal I}$ are bicolored.
\item If $\sum_{v \in V}u^1_v =  \sum_{v \in V} u^2_v$ belongs to ${\cal S}$, there is exactly one non-trivial connected component $C$ in $G({\cal S})$, all nodes of $C$ are monochromatic, and all nodes from ${\cal I}$ are bicolored, except one that is monochromatic. 
\end{enumerate}
\end{lemma}
\begin{proof}
\begin{claim}\label{cl:1}
$G({\cal S})$ does not have cycles.
\end{claim} \emph{Proof of claim:} Since $K_{r,s}$ is bipartite, any cycle $C$ in $G({\cal S})$ must be even. Alternatively summing and subtracting the equalities~\eqref{eq:ciao} corresponding to edges of $C$, we obtain $0=0$, contradicting linear independence of constraints from ${\cal S}$. \hfill $\diamond$

\begin{claim}\label{cl:2}
$G({\cal S})$ contains at least one edge.
\end{claim} \emph{Proof of claim:} Suppose the thesis does not hold. Then the only constraints in ${\cal S}$ are of the form $u_{v}=0$ and, possibly, $\sum_{v \in V}u^1_v  =  \sum_{v \in V} u^2_v$. Suppose first that $\sum_{v \in V}u^1_v  =  \sum_{v \in V} u^2_v$ is not contained in ${\cal S}$. Then there exists a node $v \in V$ that is monochromatic, while all the other nodes of $G({\cal S})$ are bicolored. Assume w.l.o.g.~that $v$ is colored red. If $u_v^1=0$, then $u$ is the zero vector, which is clearly not fundamental. Hence, $u^1_v>0$. Then the inequality corresponding to $u$ is dominated by the edge inequality corresponding to $vw$ for some $w$ in the neighborhood of $v$, a contradiction.

Hence, assume that $\sum_{v \in V}u^1_v  =  \sum_{v \in V} u^2_v$ belongs to ${\cal S}$. %Then all nodes are bicolored except $2$. 
Since $u$ is not the zero vector and ${\cal S}$ is canonical, there must be nodes $v,v'$ (possibly $v=v'$) with $u^1_{v}=u^2_{v'}>0$, while nodes from $V\setminus \{v,v'\}$ are bicolored. If $v,v'$ are on opposite sides of the vertex bipartition, then we can replace $\sum_{v \in V}u^1_v  =  \sum_{v \in V} u^2_v$ with the constraint corresponding to edge $vv'$, showing that ${\cal S}$ is not canonical. Hence, assume that $v,v'$ are on the same side of the bipartition, and assume w.l.o.g.~that $v,v' \in A$. Then the inequality corresponding to $u$ is again dominated by the edge inequality corresponding to $vw$ for some $w \in B$, a contradiction.  \hfill $\diamond$

\begin{claim}\label{cl:3}
Let $u^1_v=0$ (resp.~$u^2_v=0$) for some $v \in V$. Then $v$ is colored blue (resp.~red).
\end{claim}
\emph{Proof of claim:} Suppose w.l.o.g.~that $u^1_v=0$ but $v$ is not colored blue. By Claim~\ref{cl:2}, $G({\cal S})$ has at least one edge. If $u^1_v=0$ is linearly independent from constraints in ${\cal S}$, then we can replace some edge constraint in ${\cal S}$ with $u^1_v=0$, contradicting the canonicity of ${\cal S}$. Hence $u^1_v=0$ can be obtained as a  linear combination of constraints in ${\cal S}$. Note that a set of equations generating $u^1_v=0$ must contain either an edge constraint incident to $v$ or $\sum_{v \in V} u^1_v = \sum_{v \in V} u^2_v$ (or both), since those are the only other constraints whose support contains $u^1_v$. Hence, {we can remove from ${\cal S}$ all constraints that contain $u^1_v$, and add $u^1_v=0$. This will preserve linear independence. Then we can complete the system of tight inequalities to reach $2(|A|+|B|)-1$ linearly independent constraints and have strictly more colored nodes than ${\cal S}$, contradicting the canonicity of ${\cal S}$.}  \hfill $\diamond$

\smallskip

1. We first show that $v,w,$ cannot be colored with the same color. Suppose w.l.o.g.~both $v,w$ are  colored blue. In particular, $u^1_v=u^1_w=0$. Since $vw$ is an edge of $G({\cal S})$, we have $u^2_v + u^2_w = 0$, which by nonnegativity of $u$ implies $u^2_v=u^2_w=0$. By Claim~\ref{cl:3}, $u^1_v=0$, $u^2_v=0$, $u^1_w=0$, $u^2_w=0$ belong to ${\cal S}$. By hypothesis, $u^1_w + u^1_w = u^2_w + u^2_w$ also belongs to ${\cal S}$, contradicting the fact that ${\cal S}$ is linearly independent. 

% hypothesis, $u^1_v=0, u^1_w=0$, and $ (*) \, u^1_w + u^1_w = u^2_w + u^2_w$ belong to ${\cal S}$. Let ${\cal S}'$ be obtained from ${\cal S}$ by removing $(*)$. If both $u^2_v=0$, $u^2_w=0$ can be obtained as linear combination of equations from ${\cal S}'$, then also $(*)$ can be obtained as a linear combination of equations from ${\cal S}'$, contradicting that ${\cal S}$ is linearly independent. Hence w.l.o.g.~$u^2_w=0$ cannot be obtained as a linear combination of equations from ${\cal S}'$-- in particular, it does not belong to ${\cal S}$. Then we can replace $(*)$ with $u^2_w=0$ in ${\cal S}$ as to obtain a set still supporting $u$, hence contradicting the hypothesis that ${\cal S}$ is canonical.

We conclude the proof by showing that both $v,w$ are colored. Suppose by contradiction that $v$ is not colored. { Then $u^1_v, u^2_v > 0$ by Claim~\ref{cl:3}. Reduce $u^1_v, u^2_v$ by the minimum of the two, as to obtain vector $u'$. $u' \in Y$ since all constraints from~\eqref{eq:ciao} other than nonnegativity that contains $u_v^1$, also contain $u_v^2$ with the opposite coefficient. The inequality associated to $u$ is dominated by a conic combination of the inequality associated to $u'$ and the node inequality~\eqref{m6:c1} associated to $v$, a contradiction.}

\smallskip 

%So at least one of $u^1_v,u^2_v$, and symmetrically of $u^1_w,u^2_w$, is equal to $0$. Recall that we also have $u^1_v + u^1_w = u^2_v + u^2_w$. If $u^1_v=u^1_w=0$ or $u^2_v=u^2_w=0$, then by nonnegativity of $u$, we have that the four components are equal to $0$. Repeating an argument similar to the one above, we deduce that $u^1_v + u^1_w = u^2_v + u^2_w$ can be replaced in ${\cal S}$ by one of $u^1_v=0$, $u^2_v=0$, $u^1_w=0$, $u^2_w=0$, contradicting the fact that ${\cal S}$ is canonical. 

2. We know by Claim~\ref{cl:2} and Part 1 that $G$ has a non-trivial connected component ${C}$, and that all its nodes are monochromatic. Let $k$ be the number of nodes of $C$. By Claim~\ref{cl:1}, $C$ has $k-1$ edges, %and by Claim~\ref{cl:2}, t
and the total number of colors used in nodes from $C$ is exactly $k$.  %\revision{at most} {\color{red}this should be exactly? since each node is monochromatic by Part 1}  $k$. 
Since $\sum_{v \in V}u^1_v =  \sum_{v \in V} u^2_v$ does not belong to ${\cal S}$,  there are %at most 
{exactly} $2k-1$ constraints from ${\cal S}$ supported over some of the $2k$ variables indexed over nodes of $C$. No other non-trivial connected component of $G({\cal S})$ can then exist: since ${\cal S}$ contains $2(|V|)-1-(2k-1)=2(|V|-k)$ constraints in addition to those supported by variables indexed by nodes of $C$, every node not in $C$ must be isolated and bicolored. 

%Suppose that some node $v$ from $C$ is bicolored. Then $u^1_v=u^2_v=0$. Using the constraints on edges of $C$ and Part 2, we obtain $u^1_w = u^2_w>0$ for some node $w$ on the opposite side of the bipartition from $v$. Let $u'$ be obtained from $u$ by setting $(u')^1_w=(u')^1_w=0$. Then the inequality associated tparo $u$ is dominated by a conic combination of the inequality associated to $u'$ and inequality~\eqref{m6:c1} associated to $w$. Hence, each no node of $C$ is bicolored, and since the number of colors used by nodes of $C$ coincides with the number of vertices, we conclude that every node of $C$ is monochromatic.

3. By applying an argument similar to Part 2~above, we deduce that, for each non-trivial connected component $C$ of $G({\cal S})$ with $k$ nodes, there are %at most 
{exactly} $2k-1$ constraints from ${\cal S}$ whose support is contained on the set of variables corresponding to nodes from $C$. There is one more constraint in ${\cal S}$ that involves variables associated to nodes in $C$, and this is $\sum_{v \in V}u^1_v =  \sum_{v \in V} u^2_v$. Hence, we distinguish the following cases:
\begin{enumerate}[a)]
\item There are exactly two non-trivial connected components in $G({\cal S})$, all its nodes are monochromatic, and all nodes from ${\cal I}$ are bicolored.
%\item There is exactly one connected component in $G({\cal S})$, all its nodes are monochromatic, except one that is not colored, and all nodes from ${\cal I}$ are bicolored. {\color{red}this is to be excluded if the proof of 2.~above is completed}
\item There is exactly one non-trivial connected component in $G({\cal S})$, all its nodes are monochromatic, and all nodes from ${\cal I}$ are bicolored, except one that is monochromatic.
\end{enumerate}

We conclude the proof by showing that a) cannot happen.  Indeed, let $C_\alpha$ and $C_\beta$ be the two non-trivial connected components. Since all nodes of each non-trivial connected component are monochromatic (with nodes on the opposite sides of the same non-trivial connected component having different colors) by Part 1, all non-zero variables associated to vertices of $C_\alpha$ (resp.~$C_\beta$) have the same value $\alpha$ (resp.~$\beta$). Let $A_\alpha$, $B_\alpha$ (resp.~$A_\beta$, $B_\beta$) be the two sides of the bipartition of component $C_\alpha$ (resp.~$C_\beta$). We assume without loss of generality that $\alpha=1\leq \beta$ and $|A_\beta|\revision{\geq} |B_\beta|$. We also assume that nodes of $A_\alpha \cup A_\beta$ are red and nodes $B_\alpha \cup B_\beta$ are blue, the other cases following analogously. Note that $\sum_{v \in V}u^1_v =  \sum_{v\in V} u^2_v$ then implies 
\begin{equation}\label{eq:alpha-beta}
 |A_\alpha|+ \beta |A_\beta| = |B_\alpha| + \beta |B_\beta|,    
\end{equation}
which, together with $|A_\beta|\revision{\geq} |B_\beta|$, implies $|A_\alpha|\revision{\leq} |B_\alpha|$.

The inequality associated to $u$ is as follows %{\color{red}We previously forgot the blue term below}
\begin{align}\label{RaRb}
\sum\limits_{v \in A_{\alpha} \cup B_{\alpha}}x_v &  + \sum\limits_{v \in A_{\beta} \cup B_{\beta}}\beta x_v + \sum\limits_{e \in ((A_{\alpha} \cup A_{\beta}) \times (B_{\alpha} \cup B_{\beta})) {\setminus (A_\beta \times B_\beta)}} y_{e} \revision{+ \sum\limits_{e \in A_\beta \times B_\beta} \beta y_{e}}\leq |A_{\alpha}|+\beta|A_{\beta}|.
\end{align}
If $\beta=1$, then $|A_\alpha|+ |A_\beta| = |B_\alpha| + |B_\beta|$ and~\eqref{RaRb} is exactly the balanced biclique inequality associated to the subgraph of $K_{r,s}$ with bipartition $(A_\alpha  \cup B_\alpha, A_\beta  \cup B_\beta)$, a contradiction. So assume $\beta >1$. \revision{Then~\eqref{RaRb} is the unbalanced biclique inequality with bipartition $(A_\alpha  \cup A_\beta, B_\alpha  \cup B_\beta)$ and $A_1=A_\beta$, $B_1=B_\beta$.}\qed \end{proof}

We can now complete the proof of Theorem~\ref{thm:complete-description}. Since inequalities~\eqref{m6:c1} -- \eqref{m6:c3} and~\eqref{complete}, ~\eqref{biclique} are facet-defining for $P_T(G)$, it suffices to show that no other inequality (other than positive scalings of those) defines a facet of $P_T(G)$. It suffices to consider inequalities associated to $u$ as in Theorem~\ref{thm:projection-cone}, for $u$ fundamental, the set ${\cal S}$ supporting $u$ being canonical, and the smallest non-zero entry of $u$ being $1$. Thanks to Lemma~\ref{lem:cycle}, we have a good understanding of how such $u$, ${\cal S}$ may look like.

Assume first $\sum_{v \in V}u^1_v = \sum_{v \in V}u^2_v$ does not belong to ${\cal S}$. Then $G({\cal S})$ has exactly one non-trivial connected component $C$ (whose nodes are all monochromatic) and all other nodes are bicolored. Moreover, nodes of $C$ from opposite sides of the bipartition are colored with different colors. Similarly to the proof of Lemma~\ref{lem:cycle}, Part 3, we deduce that all non-zero components of $u$ have the same value. Hence, for some $j \in \{1,2\}$, we have $u^j_v=u^{3-j}_w=1$ and $u^{3-j}_v=u^{j}_w=0$ for all $v \in A':=A \cap V({C}), w \in B':=B \cap V({C})$. If $j=2$, the inequality associated to $u$ is 
$$
\sum_{e = \{v,w\}, v \in A', w \in B'} y_e \leq \max\{|A' |,|B' |\},
$$ while if $j=1$ the inequality we obtain is 
$$\sum_{v \in A' } x_v + \sum_{w \in B' } x_w + \sum_{e = \{v,w\}, v \in A', w \in B' } y_e \leq \max\{|A'|,|B'|\}.
$$
Both those inequalities coincide or are dominated by the balanced and non-balanced biclique inequalities associated to the pair $A',B'$. 

Now, assume that $\sum_{v \in V}u^1_v = \sum_{v \in V}u^2_v$ belongs to ${\cal S}$. Then $G({\cal S})$ has exactly one singleton monochromatic node $\overline v$, plus exactly one non-trivial connected component $C$ induced by all other monochromatic nodes.  Let $A':=A \cap V({C})$ and $B':=B \cap V({C})$. Again, nodes from $A'$ and $B'$ are colored with opposite colors and we deduce that all non-zero components of $u$ indexed by nodes of $C$ have w.l.o.g.~value $1$. W.l.o.g.~assume that $|A'| < |B'|$. Then, $\overline{v}$ is colored of the same color of nodes of $A'$, for otherwise, by using $\sum_{v \in V}u^1_v = \sum_{v \in V}u^2_v$ we obtain $u=0$, a contradiction. %Similarly as the previous case, for some $j \in \{1,2\}$, we have $u^j_v=u^{3-j}_w=1$ and $u^{3-j}_v=u^{j}_w=0$ for all $v \in A', w \in B'$. %We may assume that $\overline{v}$ is colored red, for otherwise, by using $\sum_{v \in V}u^1_v = \sum_{v \in V}u^2_v$ we obtain $u=0$, a contradiction. 
Assume also that nodes in $A'$ (resp., $B'$) are colored red (resp., blue), and $\overline{v} \in A$, the other cases following analogously. Then, by a simple computation we derive that $u^1_{\overline{v}}=|B'|-|A'|$. %If $j=2$ we conclude as the previous case. If $j=1$, 
The inequality associated to $u$ is therefore
\[
\sum\limits_{v \in A' } x_v + (|B'|-|A'|)x_{\overline{v}} + \sum\limits_{w \in B' } x_w + \sum\limits_{e = \{v,w\}, v \in A'\cup \{\bar v\}, w \in B' } y_e \leq |B'|.
\]
which is exactly the non-balanced lifted biclique inequality with $A=A'\cup \{\bar v\}$, $B=B'$, $A_1=\{v\}$, $B_1=\emptyset$. This concludes the proof.

\smallskip 

\noindent  {\bf Acknowledgments.} This research was partially supported by: the Italian Ministry of Education, University and Research (MIUR), Dipartimenti di Eccellenza Program (2018--2022) - Dept.~of Mathematics ``F. Casorati'', University of Pavia; Dept.~of Mathematics and its Applications, University of Milano-Bicocca; National Institute of High Mathematics (INdAM) ``F. Severi''; and by the  ONR award N00014-20-1-2091. Luca Ferrarini is deeply indebted to Stefano Gualandi for discussions and insightful observations on the topic. {We thank an anonymous reviewer for pointing out a mistake in a mathematical argument from a previous version of the manuscript.}

\bibliographystyle{apalike}
\bibliography{biblio}

\end{document}